\documentclass[english,15pt]{article}
\usepackage[T1]{fontenc}
\usepackage[latin1]{inputenc}
\usepackage{geometry}
\usepackage{float}
\usepackage{mathrsfs}
\usepackage{amsmath}
\usepackage{amssymb}
\usepackage{graphicx}
\usepackage{hyperref}
\usepackage[color=yellow]{todonotes}
\hypersetup{
    colorlinks=true,
    linkcolor=blue,
    filecolor=magenta,      
    urlcolor=cyan,
    citecolor = blue,
}

\usepackage{yhmath}
\makeatletter

\usepackage{bm}
\numberwithin{equation}{section}

\floatstyle{ruled}
\newfloat{algorithm}{tbp}{loa}
\providecommand{\algorithmname}{Algorithm}
\floatname{algorithm}{\protect\algorithmname}

\usepackage{algorithm}
 \usepackage{algorithmicx}

\usepackage{amsthm}

\usepackage{mathrsfs}
\usepackage{amssymb}

\usepackage{amsfonts}

\usepackage{epsfig}

\usepackage{bm}

\usepackage{mathrsfs}

\usepackage{enumerate}

\@ifundefined{definecolor}{\@ifundefined{definecolor}
 {\@ifundefined{definecolor}
 {\usepackage{color}}{}
}{}
}{}

\usepackage{subfig}\usepackage[all]{xy}

\newtheorem{prop}{Proposition}[section]

\newcounter{hypA}
\newenvironment{hypA}{\refstepcounter{hypA}\begin{itemize}
  \item[({\bf A\arabic{hypA}})]}{\end{itemize}}
\newcounter{hypB}

\newcounter{hypD}
\newenvironment{hypD}{\refstepcounter{hypD}\begin{itemize}
 \item[({\bf D\arabic{hypD}})]}{\end{itemize}}

\usepackage{babel}\date{}

\makeatother

\begin{document}

\begin{center}

{\Large \textbf{Bayesian Parameter Inference for Partially Observed Diffusions using Multilevel Stochastic Runge-Kutta Methods}}

\vspace{0.5cm}

PIERRE DEL MORAL$^{1}$, SHULAN HU$^{2}$, AJAY JASRA$^{3}$, HAMZA RUZAYQAT$^{3}$  \& XINYU WANG$^{4}$

{\footnotesize 
$^{1}$Institut de Mathematiques de Bordeaux, 33405 Bordeaux, FR\\
$^{2}$School of Mathematics and Statistics, Wenlan School of Business, \\
$^{3}$Applied Mathematics and Computational Science Program, \\ Computer, Electrical and Mathematical Sciences and Engineering Division, \\ King Abdullah University of Science and Technology, Thuwal, 23955-6900, KSA.\\
$^{4}$ Zhongnan University of Economics and Law, CN.
} \\
\footnotesize{E-Mail: \verb|pierre.del-moral@inria.fr, hu_shulan@zuel.edu.cn|,
\verb|ajay.jasra@kaust.edu.sa, hamza.ruzayqat@kaust.edu.sa, wang_xin_yu@zuel.edu.cn|}

\begin{abstract}
We consider the problem of Bayesian estimation of static parameters associated to a partially and discretely observed diffusion process. We assume that the exact transition dynamics of the diffusion process are unavailable, even up-to an unbiased estimator and that one must time-discretize the diffusion process. In such scenarios it has been shown how one can introduce the multilevel Monte Carlo method to reduce the cost to compute 
posterior expected values of the parameters
for a pre-specified mean square error (MSE); see \cite{jasra_bpe_sde}. These afore-mentioned methods rely on upon the Euler-Maruyama discretization scheme which is well-known in numerical analysis to have slow convergence properties. We adapt stochastic Runge-Kutta (SRK) methods for
Bayesian parameter estimation of static parameters for diffusions. This 
can be implemented in high-dimensions of the diffusion and seemingly under-appreciated in the uncertainty quantification and statistics fields. 
For a class of diffusions and SRK methods, we consider the estimation of the posterior expectation of the parameters.
We prove
that to achieve a MSE of $\mathcal{O}(\epsilon^2)$, for $\epsilon>0$ given, the associated work is $\mathcal{O}(\epsilon^{-2})$. 
Whilst the latter is achievable for the Milstein scheme, this method is often not applicable for diffusions in dimension larger than two. We also illustrate our methodology in several numerical examples.
\end{abstract}

\end{center}

\section{Introduction}

We are given a diffusion process on $[0,T]$, $T\in\mathbb{N}$, $\theta\in\Theta\subseteq\mathbb{R}^{d_{\theta}}$:
\begin{equation}\label{eq:diff}
dX_t = \mu_{\theta}(X_t)dt + \sigma_{\theta}(X_t)dW_t
\end{equation}
where $X_t\in\mathbb{R}^d$, $\mu:\mathbb{R}^d\times\Theta\rightarrow\mathbb{R}^d$, $\sigma:\mathbb{R}^d\times\Theta\rightarrow\mathbb{R}^{d\times d}$,
$X_0=x_0\in\mathbb{R}^d$ given and $\{W_t\}_{t\in[0,T]}$ is a standard $d-$dimensional Brownian motion. 
It is assumed that the drift and diffusion coefficients obey standard global Lipschitz and growth conditions (e.g.~\cite{kloeden}), uniform in $\theta$ to allow the existence and uniqueness of the solution $X_t$; we will make this more precise in the next section.
 We will have access to observations at unit times $Y_1,Y_2,\dots,Y_T$, $Y_k\in\mathsf{Y}$ and that conditional on $\{x_t\}_{t\in[0,T]}$, $\theta$ there is joint probabiilty density:
$$
p(y_{1:T}|\{x_t\}_{t\in[0,T]},\theta) = \prod_{k=1}^T g_{\theta}(y_k|x_k)
$$
where $g:\mathsf{Y}\times\mathbb{R}^d\times\Theta\rightarrow\mathbb{R}^+$ is a probability density function on $\mathsf{Y}$ and we have used the short hand $y_{1:T}=(y_1,\dots,y_T)$. Given a prior density, $\pi$, on the unknown parameters $\theta$, the objective is to consider computing expectations w.r.t.~the posterior density:
$$
\Pi\left(d(x_{1:T},\theta)\right) \propto \left\{\prod_{k=1}^T g_{\theta}(y_k|x_k)f_{\theta}(dx_k|x_{k-1})\right\}\pi(\theta)d\theta
$$
where $f_{\theta}$ is the transition kernel associated to \eqref{eq:diff} and $d\theta$ is the Lebesgue measure.
This problem has a wide number of applications in finance, econometrics and epidemiology; see \cite{cappe,golightly,ryder}.

In many cases of practical interest the transition density associated to $f_{\theta}$ is not available analytically and moreover, exact simulation from the process \eqref{eq:diff}, even on a discrete time grid may not be possible in general. As a result, much of the work that considers posterior inference, works directly with a time discretization of the diffusion process. Throughout the article, we consider a time-discretization on a regularly spaced grid of spacing 
$\Delta_l=2^{-l}$. Such a choice is consistent with much of the literature and is not restrictive in the sense that it can be modified with little change to the resulting methodology to be described. Once an explicit discretization scheme is chosen, one considers Bayesian inference related to a modified posterior $\Pi_l$ that we will specify later on. Perhaps the most popular form of discretization scheme that is used in uncertainty quantification (UQ) and statistics is the Euler-Maruyama (E-M) method \cite{chada_ub,golightly,jasra_bpe_sde,ub_levy, ryder}. This approach can be exactly simulated using Gaussian random variables and requiring only access to the drift and diffusion coefficients. This is then combined with a Markov chain Monte Carlo (MCMC) method and Bayesian inference is performed; see for instance the articles \cite{chada_ub,golightly,graham,jasra_bpe_sde,ub_levy,ryder}.

It is well-known that for a classes of diffusion processes, the multilevel Monte Carlo (MLMC) method of \cite{giles,hein} can be used to assist and enhance Bayesian parameter inference in the afore-mentioned problems; see \cite{jasra_bpe_sde,ml_rev}. The MLMC method consists of approximating a collapsing sum representation of
expectations w.r.t.~a sequence of posterior densities $\Pi_0,\dots,\Pi_L$. The idea is then to approximate the differences of expectations w.r.t.~$\Pi_l$ and $\Pi_{l-1}$, $l\in\{1,\dots,L\}$, by sampling an appropriate coupling of $(\Pi_l,\Pi_{l-1})$. If the coupling is suitably well-defined and an appropriate difference between $\Pi_L$ and $\Pi$ is sufficiently well understood, it is possible to choose the Monte Carlo samples approximating the differences between 
$\Pi_l$ and $\Pi_{l-1}$ in such a way as to improve upon simply considering $\Pi_L$. More precisely, under assumptions, in \cite{jasra_bpe_sde} using a type of importance sampling identity combined with MCMC, when using the E-M method
the cost to achieve a mean square error (MSE) of $\mathcal{O}(\epsilon^2)$, for $\epsilon>0$ given, is $\mathcal{O}(\epsilon^{-2}\log(\epsilon)^2)$. If one considered only an MCMC algorithm associated to $\Pi_L$ then this cost increases to $\mathcal{O}(\epsilon^{-3})$.

Fundamental to the MLMC method of \cite{jasra_bpe_sde} is the time-discretization that is used as well as the weak and strong error rates; see e.g.~\cite{kloeden}. By weak error we mean
$$
|\mathbb{E}\left[X_T-X_T^l\right]| 
$$
where $X_T^l$ is a solution associated of a time discretization method, with time step $\Delta_l$, $\mathbb{E}$ is an expectation operator on an appropriately defined probability space and $|\cdot|$ is the $\mathbb{L}_1-$norm. The strong error is taken as
$$
\mathbb{E}\left[\|X_T-X_T^l\|^2\right]
$$
where $\|\cdot\|$ is the $\mathbb{L}_2-$norm. It is well-known that for the E-M method, under assumptions, the weak error is $\mathcal{O}(\Delta_l)$ (i.e.~order one, referring to the exponent of $\Delta_l$) and the strong error rate is $\mathcal{O}(\Delta_l)$ also order 1; see e.g.~\cite{kloeden}. These two rates are important as the weak error is essentially used to determine the bias, i.e.~the difference between $\Pi_L$ and $\Pi$ and the strong error is used to measure the degree of (so-called synchronous) coupling between $\Pi_l$ and $\Pi_{l-1}$; although the arguments are more sophisticated than that, this is the essence of them. If the strong or weak error can increase in order then the associated computational cost can be reduced to achieve a given MSE; as we will use MCMC, the optimal such cost is $\mathcal{O}(\epsilon^{-2})$ to obtain an MSE of $\mathcal{O}(\epsilon^{2})$.

The popularity of the E-M method in UQ and statistics can be stemmed from the fact that it is easy to simulate (depending on increments of Brownian motion) and does not require any derivative information associated to the drift and diffusion coefficients. Such a method is then relatively fast, especially when needed repeatedly as would be the case when using MCMC.
None-the-less as has been frequently mentioned in the numerical analysis literature, the E-M method is very slow in terms of its convergence; see e.g.~\cite{bur_bur,kloeden,rossler,rum}. An alternative that is well-known and that has been used in the UQ and statistics literature (e.g.~\cite{jasra_anti}) along with multilevel methods is the Milstein approach. This technique corresponds with E-M if $\sigma_{\theta}$ in \eqref{eq:diff} is a constant, otherwise it often needs L\'evy areas when $d$ is at least two or more; this can be circumvented when using derivative information as in \cite{ml_anti}.  However, when applicable, this method can in principle achieve an optimal cost for a given MSE and is hence perhaps the gold standard that is used in UQ and statistics. This derives from the fact that the strong error is order two.

In this article we consider using the well-known class of stochastic Runge-Kutta (SRK) methods. These are a class of discretization methods that have been extensively studied in numerical analysis, since at least the 1970s and of which the E-M and Milstein schemes are special cases; we refer to some of the key papers in that field for a formal review \cite{bur_bur,rossler,rossler1,rum}. These approaches, as one might expect, are based upon the well-known Runge-Kutta method for the numerical solution of ordinary differential equations (ODEs). SRK methods, when they are explicit, are often of a similar form as the E-M or Milstein scheme and have a cost that (in terms of order) is linear in the number of discretization points. Moreover, they can sometimes only be based upon increments of Brownian motion.
The main point, however, is that the strong and sometimes even weak error \cite{tocino} is increased versus the Milstein scheme and is fully implementable even in dimensions bigger than two. For instance, if one only relies on Brownian increments, as we shall in this article, a well-known result in the literature is that the strong error of SRK methods has a strong error of at most three, which can be achievable; see e.g.~\cite{bur_bur,rum}.

In this article we develop MCMC methods for approximating a multilevel identity associated to $\Pi_L$, where the latter is defined via an SRK discretization of the process \eqref{eq:diff}. As noted in the abstract, it seems that these SRK discretization methods are largely under-appreciated in the UQ and statistics literature and one of the purposes of this article is to promote the usage of such methodology. The main contributions of this article then are as follows:
\begin{enumerate}
\item{An MCMC method for approximating a multilevel identity associated to $\Pi_L$, where the latter is defined via an SRK discretization of the process \eqref{eq:diff}.}
\item{For this method we prove that for a class of diffusions and under some additional assumptions that to achieve a MSE, associated to posterior expectations, of $\mathcal{O}(\epsilon^2)$, for $\epsilon>0$ given, the associated work is $\mathcal{O}(\epsilon^{-2})$.}
\item{We illustrate and compare our methodology in several numerical examples.}
\end{enumerate}
At this stage, it is worth emphasizing that the approach here, depends only on Brownian increments and has a cost to simulate that is of the same order as the E-M and Milstein schemes. It does require access to the first derivatives, in the $x$ variable of $\sigma_{\theta}(x)$ and this is due to a correction term. If one considers the Stratanovich form of \eqref{eq:diff}, then one does not need access to derivatives of $\sigma_{\theta}(x)$.  Although SRK methods may be under-appreciated in the statistics and UQ literature several higher-order (in terms of weak or strong errors) methods have been considered for hypo-elliptic diffusions \cite{ditvelson,iguchi}. Indeed \cite{iguchi} advocate their methodology as a general discretization scheme for (elliptic and hypo-elliptic) diffusion processes. We remark however, that whilst such a method may be useful in general, to the best of our knowledge there are no strong error results and the error orders can be improved in several contexts. Personal communication with the authors \cite{iguchi} suggest that the strong order will be 1 under the assumptions in that paper, but at present there is no proof.

This article is structured as follows. In Section \ref{sec:approach} we detail the approach that is used in this article. 
In Section \ref{sec:math} our mathematical results are given. In Section \ref{sec:numerics} we present several numerical results.

%

\section{Approach}\label{sec:approach}

Throughout this article we assume the following.
\begin{hypD}\label{hyp:d1}
\begin{itemize}
\item{$\mu_{\theta}$ and $\sigma_{\theta}$ are globally Lipschitz functions: there exists a $C<+\infty$ such that for any $(x,x',\theta)\in\mathbb{R}^{2d}\times\Theta$
$$
\|\mu_{\theta}(x)-\mu_{\theta}(x')\| + \|\sigma_{\theta}(x)-\sigma_{\theta}(x')\| \leq C\|x-y\|
$$
where if $x\in\mathbb{R}^d$, $\|x\|$ is the $\mathbb{L}_2-$norm and if $x\in\mathbb{R}^{d\times d}$
$\|x\|^2=\sum_{(i,j)\in\{1,\dots,d\}^2}x_{ij}^2$.}
\item{$\mu_{\theta}$ and $\sigma_{\theta}$ satisfy a growth condition: there exists a $C<+\infty$ such that for 
any $(x,\theta)\in\mathbb{R}^{d}\times\Theta$
$$
\|\mu_{\theta}(x)\| +  \|\sigma_{\theta}(x)\| \leq C\left(1+\|x\|\right).
$$
}
\item{For each $\theta\in\Theta$, $\sigma_{\theta}(x)$ is differentiable in $x$.}
\end{itemize}
\end{hypD}
We note that this assumption may be strengthened at various points in the article, but it is taken as a minimum requirement.

\subsection{Discretization Scheme}\label{sec:disc_scheme}

To describe our approach we introduce the following function, for $x\in\mathbb{R}^d$:
$$
\overline{\mu}_{\theta}(x) := \mu_{\theta}(x) - \lambda \overline{\sigma}_{\theta}(x)
$$ 
where $\lambda\in\mathbb{R}$ is a constant to be determined below and 
$\overline{\sigma}_{\theta}(x)$ is a $d\times 1$ vector where for $i\in\{1,\dots,d\}$
$$
\overline{\sigma}_{\theta}^{(i)}(x) = \sum_{(p,j)\in\{1,\dots,d\}^2}\frac{\partial\sigma_{\theta}^{(ip)}}{\partial x_j}(x)
\sigma_{\theta}^{(jp)}(x)
$$
with $x_j$ the $j^{th}$ element of the $d-$vector $x$, $\overline{\sigma}_{\theta}^{(i)}$ denoting the $i^{th}-$element of $\overline{\sigma}_{\theta}$ and
$\sigma_{\theta}^{(ip)}$ the $(i,p)^{th}-$element of $\sigma_{\theta}$.

We focus our attention on general $s-$stage SRK methods ($s\in\mathbb{N}$) for \eqref{eq:diff} to be given below (see e.g.~\cite[Eq.~(26)]{bur_bur}). Let $A$ and $B$ be $s\times s$ strictly lower triangular matrices, $\alpha$ and $\gamma$ be $s\times 1$ vectors with
$$
\sum_{p=1}^{s}\alpha_p=\sum_{p=1}^{s}\gamma_p=1
$$
and $\lambda = \gamma^{\top}Be$ where $e$ is the $d\times 1$ vector of ones. Then for $(k,j)\in\{0,\dots, T\Delta_l^{-1}-1\}\times\{1,\dots,s\}$ we consider the discretization
\begin{eqnarray}
X_{(k+1)\Delta_l}^l & = & X_{k\Delta_l}^l + \Delta_l\sum_{p=1}^s \alpha_p\overline{\mu}_{\theta}(V_{k\Delta_l}^l(p)) + 
\left\{\sum_{p=1}^s \gamma_p\sigma_{\theta}(V_{k\Delta_l}^l(p))\right\}\left[W_{(k+1)\Delta_l} - W_{k\Delta_l}\right] 
\label{eq:x_update}
\\
V_{k\Delta_l}^l(j) & = & X_{k\Delta_l}^l + \Delta_l\sum_{p=1}^sa_{jp}\overline{\mu}_{\theta}(V_{k\Delta_l}^l(p)) +
\left\{\sum_{p=1}^s b_{jp}\sigma_{\theta}(V_{k\Delta_l}^l(p))\right\}\left[W_{(k+1)\Delta_l} - W_{k\Delta_l}\right].
\label{eq:v_update}
\end{eqnarray}
It is remarked that this discretization is by no means the most sophisticated or the best in terms of strong/weak errors, but are relatively fast to compute, which the main point. To illustrate, we present the 
stochastic Heun (SH) method e.g.~\cite[Eq.~(2.5)]{rum}, which is described in Algorithm \ref{alg:rkm} and a stochastic classical
Runge-Kutta scheme in Algorithm \ref{alg:rkm1}. We note, as we will discuss in Section \ref{sec:wz} using classical ODE-type solvers is not necessarily the `optimal' way to proceed, but Algorithm \ref{alg:rkm1} has good strong error properties under some assumptions as we now discuss.

It is shown in various articles e.g.~\cite{rum} and under mathematical assumptions that one has
$$
\mathbb{E}\left[\|X_T-X_T^l\|^2\right] =\mathcal{O}(\Delta_l^{\beta})
$$
where $\beta>0$. For instance, under assumptions, \cite[Theorem 3 (ii)]{rum} states that $\beta=3$ for the SH method in Algorithm \ref{alg:rkm}. When $d=1$, under assumptions, \cite[Theorem 2 (ii)]{rum} shows that $\beta=4$ in the SRK method of Algorithm \ref{alg:rkm1}.  

\begin{algorithm}
\begin{enumerate}
\item{Input: Level of discretization $l$, initial point $X_0^l=x_0^l$, final time $T\in\mathbb{N}$, path of Brownian motion $W_0,W_{\Delta_l},\dots,W_{T}$ and parameter $\theta$.}
\item{Initialize: Set $X_0^l=x_0$, $k=0$.}
\item{Compute the coefficients ($\lambda=\tfrac{1}{2}$):
\begin{eqnarray*}
V_{k\Delta_l}^l(1) & = & X_{k\Delta_l}^l \\
V_{k\Delta_l}^l(2) & = & X_{k\Delta_l}^l + \Delta_l\overline{\mu}_{\theta}(V_{k\Delta_l}^l(1)) + \sigma_\theta(V_{k\Delta_l}^l(1))\left[W_{(k+1)\Delta_l} - W_{k\Delta_l}\right]. 
\end{eqnarray*}
}
\item{Update:
\begin{eqnarray*}
X_{(k+1)\Delta_l}^l & = & X_{k\Delta_l}^l + \frac{1}{2}\left(
\left\{\sum_{p=1}^2\overline{\mu}_{\theta}(V_{k\Delta_l}^l(p))\right\}
\Delta_l + 
\left\{\sum_{p=1}^2\sigma_{\theta}(V_{k\Delta_l}^l(p))\right\}
\left[W_{(k+1)\Delta_l} - W_{k\Delta_l}\right]\right)\\
k & = & k+1.
\end{eqnarray*}
If $k=T\Delta_l^{-1}-1$ go to 5.~otherwise return to 3..}
\item{Return the solution $X_T^l$.}
\end{enumerate}
\caption{Stochastic Heun Method}
\label{alg:rkm}
\end{algorithm}

\begin{algorithm}
\begin{enumerate}
\item{Input: Level of discretization $l$, initial point $X_0^l=x_0^l$, final time $T\in\mathbb{N}$, path of Brownian motion $W_0,W_{\Delta_l},\dots,W_{T}$ and parameter $\theta$.}
\item{Initialize: Set $X_0^l=x_0$, $k=0$.}
\item{Compute the coefficients ($\lambda=\tfrac{1}{2}$):
\begin{eqnarray*}
V_{k\Delta_l}^l(1) & = & X_{k\Delta_l}^l \\
V_{k\Delta_l}^l(j) & = & X_{k\Delta_l}^l + \tfrac{1}{2}\left(\Delta_l\overline{\mu}_{\theta}(V_{k\Delta_l}^l(j-1)) + \sigma_\theta(V_{k\Delta_l}^l(j-1))\left[W_{(k+1)\Delta_l} - W_{k\Delta_l}\right]\right) \quad j\in\{2,3\} \\
V_{k\Delta_l}^l(4) & = & X_{k\Delta_l}^l + \Delta_l\overline{\mu}_{\theta}(V_{k\Delta_l}^l(3)) + \sigma_\theta(V_{k\Delta_l}^l(3))\left[W_{(k+1)\Delta_l} - W_{k\Delta_l}\right]
\end{eqnarray*}
}
\item{Update: $\gamma=\alpha=(\tfrac{1}{6},\tfrac{2}{6},\tfrac{2}{6},\tfrac{1}{6})^{\top}$
\begin{eqnarray*}
X_{(k+1)\Delta_l}^l & = & X_{k\Delta_l}^l + \Delta_l\sum_{p=1}^4 \alpha_p\overline{\mu}_{\theta}(V_{k\Delta_l}^l(p)) + 
\left\{\sum_{p=1}^4 \gamma_p\sigma_{\theta}(V_{k\Delta_l}^l(p))\right\}\left[W_{(k+1)\Delta_l} - W_{k\Delta_l}\right] 
 \\
k & = & k+1.
\end{eqnarray*}
If $k=T\Delta_l^{-1}-1$ go to 5.~otherwise return to 3..}
\item{Return the solution $X_T^l$.}
\end{enumerate}
\caption{Classic Runge Kutta Method}
\label{alg:rkm1}
\end{algorithm}

From herein, we will denote the transition over unit time that is induced by the discretization as $f_{\theta}^l(dx_{k+1}|x_k)$. That is, by $\Delta_l^{-1}$ iterations of \eqref{eq:x_update}, which requires the use of $s-$applications of \eqref{eq:v_update} at each update.


\subsubsection{A Note on Wong-Zakai Approximation}\label{sec:wz}

We suppose $d=1$ in this section.
For $k\in\{0,\Delta_l,\dots,T\Delta_l^{-1}-1\}$, $t\in[k\Delta_l,(k+1)\Delta_l]$ set
$$
W_t^l := W_{k\Delta} + \tfrac{t-k\Delta_l}{\Delta_l}\left(W_{(k+1)\Delta_l}-W_{k\Delta_l}\right)
$$
and this implies that the time derivative is
$$
\dot{W}_t^l = \tfrac{1}{\Delta_l}\left(W_{(k+1)\Delta_l}-W_{k\Delta_l}\right).
$$
Now, under regularity conditions, it is well-known \cite{wong} that the solution to the random ODE ($\lambda=\tfrac{1}{2}$):
\begin{equation}\label{eq:rode}
d\widetilde{X}_t^l = \overline{\mu}_{\theta}(\widetilde{X}_t^l)dt + \sigma_{\theta}(\widetilde{X}_t^l)\dot{W}_t^l dt
\end{equation}
converges almost surely, as $l\rightarrow\infty$, to the solution of \eqref{eq:diff}; this is the well-known method of Wong-Zakai approximation. 

This suggests that one can also apply classical ODE solvers to approximate the transition dynamics of \eqref{eq:diff}.
As has been noted in the literature (e.g.~\cite{bur_bur}), naively using well-known ODE methods may not work well for diffusion problems. In addition, in some cases, considering numerical methods for the above random ODE can (see e.g.~\cite[Section 2.3]{rum}) lead to the incorrect solution being approximated. Indeed, \cite{rum} concludes that in this context `one always operates with some Runge Kutta scheme $\dots$ which needs its special correction'.  Thus, in general, it seems prudent to follow the SRK framework that we have adopted.

\subsection{Model of Interest and Multilevel Approximation}

Given the above exposition, our target of interest is to consider the posterior (probability) measure for some fixed $l\in\mathbb{N}$:
$$
\Pi_l\left(d(x_{1:T},\theta)\right) \propto \left\{\prod_{k=1}^T g_{\theta}(y_k|x_k)f_{\theta}^l(dx_k|x_{k-1})\right\}\pi(\theta)d\theta
$$
where $x_0^l=x_0$.

To describe our method, we first recapitulate the MLMC method and the approach that is detailed in \cite{jasra_bpe_sde}. Let $\varphi:\mathbb{R}^{Td}\times\Theta\rightarrow\mathbb{R}$ be a bounded and continuous
function ($\mathcal{C}_b(\mathbb{R}^{Td}\times\Theta)$) and for $l\in\mathbb{N}$ fixed, we write 
$$
\Pi_l(\varphi):=\int_{\mathbb{R}^{Td}\times\Theta}\varphi(x_{1:T},\theta)\Pi_l\left(d(x_{1:T},\theta)\right)
$$
respectively we write
$$
\Pi(\varphi) :=\int_{\mathbb{R}^{Td}\times\Theta}\varphi(x_{1:T},\theta)\Pi\left(d(x_{1:T},\theta)\right).
$$
The objective is to approximate $\Pi(\varphi)$ for a wide variety of $\varphi\in\mathcal{C}_b(\mathbb{R}^{Td}\times\Theta)$, given that we have determined that we can at best work with $\Pi_l$. Note that our restriction to continuous and bounded functions is nothing more than convention to keep the technical discussion to a minimum - the methodology described here is easily extended to appropriately integrable functions.

It is well-known (e.g.~\cite{jasra_bpe_sde}) that in some scenarios that instead of considering simply $\Pi_l$ one considers a sequence $\Pi_0,\dots,\Pi_L$ for some $L$ to be determined. Then we consider the identity:
\begin{equation}\label{eq:ml_id}
\Pi_L(\varphi) = \Pi_0(\varphi) + \sum_{l=1}^L\left\{\Pi_l(\varphi)-\Pi_{l-1}(\varphi)\right\}.
\end{equation}
The objective now is to develop a simulation-based method to approximate the R.H.S.~of \eqref{eq:ml_id}.

Approximating $\Pi_0(\varphi)$ is possible using an MCMC method and we detail this in the next section. The approximation of the difference $\Pi_l(\varphi)-\Pi_{l-1}(\varphi)$ requires a little more effort and we begin by defining a coupling of the discretization scheme that is detailed in Section \ref{sec:disc_scheme} across two levels $l$ and $l-1$; this is presented in Algorithm \ref{alg:coup_meth}. We denote, for $k\in\{0,\dots,T-1\}$, by 
$\check{f}_{\theta}^l\left(d(x_{k+1}^l,x_{k+1}^l))|(x_k^l,x_k^{l-1})\right)$ the transition induced by the coupling described in Algorithm \ref{alg:coup_meth}, over a unit time ($T=1$) with initial points $(x_k^l,x_k^{l-1})$.
We note that it easily follows that for any $(x_k^l,x_k^{l-1},\theta)\in\mathbb{R}^{2d}\times\Theta$ and Borel set $C$ on $\mathbb{R}^{d}$ we have
\begin{eqnarray}
\int_{C} f^l_{\theta}(dx|x_k^l) & = & \int_{C\times\mathbb{R}^d} \check{f}_{\theta}^l\left(d(x_{k+1}^l,x_{k+1}^l))|(x_k^l,x_k^{l-1})\right) \label{eq:trans_marginal1}\\
\int_{C} f^{l-1}_{\theta}(dx|x_k^{l-1}) & = & \int_{\mathbb{R}^d\times C} \check{f}_{\theta}^l\left(d(x_{k+1}^l,x_{k+1}^l))|(x_k^l,x_k^{l-1})\right). \label{eq:trans_marginal2}
\end{eqnarray}
Set for $(x,x')\in\mathbb{R}^{2d}$ and $k\in\{1,\dots,T\}$
$$
\check{g}_{k,\theta}(x,x') = \max\left\{g_{\theta}(y_k|x),g_{\theta}(y_k|x')\right\}.
$$

We consider the following probability measure on $\mathbb{R}^{2dT}\times\Theta$, $l\in\mathbb{N}$:
$$
\check{\Pi}_{l,l-1}\left(d(x_{1:T}^l,x_{1:T}^{l-1},\theta)\right) \propto 
\left\{\prod_{k=1}^T\check{g}_{k,\theta}(x_k^l,x_k^{l-1})\check{f}_{\theta}^l\left(d(x_{k}^l,x_k^{l-1}))|(x_{k-1}^l,x_{k-1}^{l-1})\right)\right\}\pi(\theta)d\theta
$$
where $(x_0^l,x_0^{l-1})=(x_0,x_0)$. For $\varphi\in\mathcal{C}_b(\mathbb{R}^{2Td}\times\Theta)$, we use the notation
$$
\check{\Pi}_{l,l-1}\left(\varphi\right) = \int_{\mathbb{R}^{2Td}\times\Theta}\varphi(x_{1:T}^l,x_{1:T}^{l-1},\theta)\check{\Pi}_{l,l-1}\left(d(x_{1:T}^l,x_{1:T}^{l-1},\theta)\right).
$$
If $\varphi\in\mathcal{C}_b(\mathbb{R}^{Td}\times\Theta)$ we write for $(x_{1:T},x_{1:T}')\in\mathbb{R}^{2Td}$
\begin{eqnarray*}
\varphi \otimes 1(x_{1:T},x_{1:T}',\theta) & = & \varphi(x_{1:T},\theta) \\
1\otimes \varphi(x_{1:T},x_{1:T}',\theta) & = & \varphi(x_{1:T}',\theta).
\end{eqnarray*}
Set for $(x_{1:T},x_{1:T}')\in\mathbb{R}^{2Td}$
\begin{eqnarray*}
\check{H}_{1}(x_{1:T},x_{1:T}',\theta) & = & \prod_{k=1}^T \frac{g_{\theta}(y_k|x_k)}{\check{g}_{k,\theta}(x_k,x_k')} \\
\check{H}_{2}(x_{1:T},x_{1:T}',\theta) & = & \prod_{k=1}^T \frac{g_{\theta}(y_k|x_k')}{\check{g}_{k,\theta}(x_k,x_k')}.
\end{eqnarray*}
Then just as in \cite{jasra_bpe_sde} one has that
$$
\Pi_l(\varphi)-\Pi_{l-1}(\varphi) = \frac{\check{\Pi}_{l,l-1}\left(\{\varphi\otimes 1\} \check{H}_{1}\right)}{\check{\Pi}_{l,l-1}\left(\check{H}_{1}\right)} - \frac{\check{\Pi}_{l,l-1}\left(\{1\otimes \varphi\} \check{H}_{2}\right)}{\check{\Pi}_{l,l-1}\left(\check{H}_{2}\right)}
$$
which we note easily extends to the context here as a result of \eqref{eq:trans_marginal1}-\eqref{eq:trans_marginal2}.

To summarize, to approximate \eqref{eq:ml_id}, one needs a method to 
approximate expectations w.r.t.~$\Pi_0$ and $\check{\Pi}_{l,l-1}$, $l\in\mathbb{N}$ and this is what we consider in the next section.

\begin{algorithm}[H]
\begin{enumerate}
\item{Input: Level of discretization $l\in\mathbb{N}$, initial points $(X_0^l,X_0^{l-1})=(x_0^l,x_0^{l-1})$,
final time $T\in\mathbb{N}$, path of Brownian motion $W_0,W_{\Delta_l},\dots,W_{T}$ and parameter $\theta$.}
\item{Run the discretization scheme in \eqref{eq:x_update}-\eqref{eq:v_update} at level $l$, with initial point $x_0^l$, the given path of Brownian motion and up-to time $T$.}
\item{Run the discretization scheme in \eqref{eq:x_update}-\eqref{eq:v_update} at level $l-1$, with initial point $x_0^{l-1}$, the given path of Brownian motion (same as step 2.) and up-to time $T$.}
\item{Return the solutions $(X_T^l,X_T^{l-1})$.}
\end{enumerate}
\caption{Coupling of Discretization Scheme in Eqs.~\eqref{eq:x_update}-\eqref{eq:v_update}}
\label{alg:coup_meth}
\end{algorithm}

\subsection{MCMC Method}

We now describe the two MCMC methods we use to approximate expectations w.r.t.~$\Pi_0$ and $\check{\Pi}_{l,l-1}$, $l\in\mathbb{N}$. We focus on using the MCMC method in \cite{andrieu}, but e.g.~an adaptation to the approach of \cite{graham} is also possible.

The approximation of $\Pi_0$ (and indeed is easily extended to any $\Pi_l$ for $l\in\mathbb{N}$) is described in Algorithm \ref{alg:mcmc_0} which needs the assistance of Algorithm \ref{alg:pf_0}. 
The notation in  Algorithm \ref{alg:mcmc_0}, $X_{1:T}^k(0)$, refers to the sampled trajectory of the diffusion, from time $1$ to time $T$, at iteration $k$ of the MCMC algorithm and the argument $(0)$ refers to the fact that one is sampling from the posterior with discretization level $0$.
Algorithm \ref{alg:mcmc_0} is simply the particle marginal Metropolis-Hastings (M-H) algorithm in \cite{andrieu} and the choice of $M$ is discussed in many articles, including \cite{andrieu}; typically $M=\mathcal{O}(T)$ is a sensible choice, which is the one we adopt.
The proposal for M-H algorithms has been discussed extensively in the literature; see e.g.~\cite{robert}.
Now, for $\varphi\in\mathcal{C}_b(\mathbb{R}^{Td}\times\Theta)$, to approximate $\Pi_0(\varphi)$ one has the following estimator:
$$
\Pi_0^N(\varphi)  = \frac{1}{N+1}\sum_{k=0}^N\varphi(X_{1:T}^{k}(0),\theta^k(0)).
$$

\begin{algorithm}[h!]
\begin{enumerate}
\item{Input: $M\in\mathbb{N}$ number of particles, $N$ number of iterations and a positive proposal kernel $q(\theta'|\theta)$ that can be simulated.}
\item{Initialize: Generate $\theta^0(0)$ from the prior and call Algorithm \ref{alg:pf_0} with parameter $\theta^{0}(0)$, $M$ number of particles. Denote the returned trajectory as $X_{1:T}^{0}(0)$ and the estimated log-normalizing constant as $p^{0,M}(y_{1:T})$. Set $k=0$}
\item{Iterate: Propose a new $\theta'$ using the $q(\cdot|\theta^k(0))$ and call Algorithm \ref{alg:pf_0} with parameter $\theta'$, $M$ number of particles and denote $p^{M,'}(y_{1:T})$ the returned value of the estimated log-normalizing constant and $X_{1:T}'$ the returned trajectory. Compute
$$
R = \log\left(\frac{\pi(\theta')q(\theta^k(0)|\theta')}{\pi(\theta^k(0))q(\theta'|\theta^k(0))}\right) + p^{0,M,'}(y_{1:T}) - p^{k,M}(y_{1:T})
$$
and generate $U\sim\mathcal{U}_{[0,1]}$ (uniform distribution on $[0,1]$). If $\log(U)<R$ set $\theta^{k+1}(0)=\theta'$, $p^{k+1,M}(y_{1:T})=p^{M,'}(y_{1:T})$, 
and $X_{1:T}^{k+1}(0) = X_{1:T}'$. Otherwise set $\theta^{k+1}(0)=\theta^{k}(0)$, $p^{k+1,M}(y_{1:T})=p^{k,M}(y_{1:T})$ 
and $X_{1:T}^{k+1}(0) = X_{1:T}^{k}(0)$. If $k=N$ go to step 4.~otherwise set $k=k+1$ and go to the start of step 3..}
\item{Return $\left((X_{1:T}^{0}(0),\theta^0(0)),\dots,(X_{1:T}^{N}(0),\theta^N(0))\right)$.}
\end{enumerate}
\caption{MCMC method for $\Pi_0$}
\label{alg:mcmc_0}
\end{algorithm}

\begin{algorithm}[h!]
\begin{enumerate}
\item{Input: $\theta$, $M\in\mathbb{N}$, initial point $x_0$ and final time $T$.}
\item{Initialize: Generate $X_1^1,\dots,X_1^M$ using $f^0_{\theta}(\cdot|x_0)$ and set $k=1$, $p^M(y_{1:-1}) = 0$.}
\item{Compute the normalize weights, for $i\in\{1,\dots,M\}$
$$
U_k^i = \frac{g_{\theta}(y_k|x_k^i)}{\sum_{j=1}^Mg_{\theta}(y_k|x_k^j)}
$$
and update $p^M(y_{1:k}) = p^M(y_{1:k-1}) + \log\left(\tfrac{1}{M}\sum_{j=1}^M g_{\theta}(y_k|x_k^j)\right)$. If $k=T$ go to 6., otherwise go to the next step.}
\item{Sample, with replacement (resampling) from $(x_{1:k}^1,\dots,x_{1:k}^M)$ using the normalized weights $(U_k^1,\dots,U_k^M)$ and denote the new samples as
$(x_{1:k}^1,\dots,x_{1:k}^M)$.}
\item{Sample, for $i\in\{1,\dots,M\}$, $X_{k+1}^i|x_k^i$ using  $f^0_{\theta}(\cdot|x_k^i)$ and set $k=k+1$ and go to step 3..}
\item{Pick a single trajectory $X_{1:T}^i$ using the normalized weights and return this and $p^M(y_{1:T})$.}
\end{enumerate}
\caption{Particle Filter (PF)}
\label{alg:pf_0}
\end{algorithm}

We now consider approximating expectations w.r.t.~$\check{\Pi}_{l,l-1}$, $l\in\mathbb{N}$. This is summarized in 
Algorithm \ref{alg:mcmc_l} which needs the assistance of Algorithm \ref{alg:pf_l} (the Delta particle filter; see \cite{chada_ub,jasra_bpe_sde,ub_levy}). Now, for $\varphi\in\mathcal{C}_b(\mathbb{R}^{2Td}\times\Theta)$, to approximate $\check{\Pi}_{l,l-1}(\varphi)$ one has the following estimator:
$$
\check{\Pi}_{l,l-1}^N(\varphi)  = \frac{1}{N+1}\sum_{k=0}^N\varphi(X_{1:T}^{k,l}(l),X_{1:T}^{k,l-1}(l),\theta^k(l)).
$$

We are now in a position to describe our approach to approximate the R.H.S.~of \eqref{eq:ml_id}. We need choices of $L$ and the number of iterations at levels $0,\dots,L$, which we denote by $N_0,\dots,N_L$. We will show how these can be chosen in the next section. The steps of the main approach are given in Algorithm \ref{alg:main_alg}. 

\begin{algorithm}[h!]
\begin{enumerate}
\item{Input level $l\in\mathbb{N}$, $M\in\mathbb{N}$ number of particles, $N$ number of iterations and a positive proposal kernel $q(\theta'|\theta)$ that can be simulated.}
\item{Initialize: Generate $\theta^0(l)$ from the prior and call Algorithm \ref{alg:pf_l} with parameter $\theta^{0}(l)$, $M$ number of particles. Denote the returned trajectory as $(X_{1:T}^{0,l}(l),X_{1:T}^{0,l-1}(l))$ and the estimated log-normalizing constant as $\check{p}^{0,M}(y_{1:T})$. Set $k=0$}
\item{Iterate: Propose a new $\theta'$ using the $q(\cdot|\theta^k(l))$ and call Algorithm \ref{alg:pf_l} with parameter $\theta'$, $M$ number of particles and denote $\check{p}^{M,'}(y_{1:T})$ the returned value of the estimated log-normalizing constant and $(X_{1:T}^{',l},X_{1:T}^{',l-1})$ the returned trajectory. Compute
$$
R = \log\left(\frac{\pi(\theta')q(\theta^k(0)|\theta')}{\pi(\theta^k(0))q(\theta'|\theta^k(0))}\right) + \check{p}^{M,'}(y_{1:T}) - \check{p}^{k,M}(y_{1:T})
$$
and generate $U\sim\mathcal{U}_{[0,1]}$. If $\log(U)<R$ set $\theta^{k+1}(l)=\theta'$, $\check{p}^{k+1,M}(y_{1:T})=\check{p}^{M,'}(y_{1:T})$
and $(X_{1:T}^{k+1,l}(l),X_{1:T}^{k+1,l-1}(l)) = (X_{1:T}^{',l},X_{1:T}^{',l-1})$. Otherwise set $\theta^{k+1}(0)=\theta^{k}(0)$, $\check{p}^{k+1,M}(y_{1:T})=\check{p}^{M,k}(y_{1:T}) $
and $(X_{1:T}^{k+1,l}(l),X_{1:T}^{k+1,l-1}(l)) =(X_{1:T}^{k,l}(l),X_{1:T}^{k,l-1}(l))$. If $k=N$ go to step 4.~otherwise set $k=k+1$ and go to the start of step 3..}
\item{Return $\left((X_{1:T}^{0,l}(l),X_{1:T}^{0,l-1}(l),\theta^0(0)),\dots,(X_{1:T}^{N,l}(l),X_{1:T}^{N,l-1}(l),\theta^N(0))\right)$.}
\end{enumerate}
\caption{MCMC method for $\check{\Pi}_{l,l-1}$}
\label{alg:mcmc_l}
\end{algorithm}

\begin{algorithm}[h!]
\begin{enumerate}
\item{Input level $l\in\mathbb{N}$, $\theta$, $M\in\mathbb{N}$, initial point $x_0$ and final time $T$.}
\item{Initialize: Generate $(X_1^{1,l},X_1^{1,l-1}),\dots,(X_1^{M,l},X_1^{M,l-1})$ using $\check{f}^l_{\theta}\left(\cdot|(x_0,x_0)\right)$ and set $k=1$, $\check{p}^M(y_{1:-1}) = 0$.}
\item{Compute the normalize weights, for $i\in\{1,\dots,M\}$
$$
\check{U}_k^i = \frac{\check{g}_{k,\theta}(x_k^{i,l},x_k^{i,l-1})}{\sum_{j=1}^M \check{g}_{k,\theta}(x_k^{j,l},x_k^{j,l-1})}
$$
and update $\check{p}^M(y_{1:k}) = \check{p}^M(y_{1:k-1}) + \log\left(\tfrac{1}{M}\sum_{j=1}^M 
\check{g}_{k,\theta}(x_k^{j,l},x_k^{j,l-1})\right)$. If $k=T$ go to 6., otherwise go to the next step.}
\item{Resample $\left((x_{1:k}^{1,l},x_{1:k}^{1,l-1}),\dots,(x_{1:k}^{M,l},x_{1:k}^{M,l-1})\right)$ using the normalized weights $(\check{U}_k^1,\dots,\check{U}_k^M)$ and denote the new samples as
$\left((x_{1:k}^{1,l},x_{1:k}^{1,l-1}),\dots,(x_{1:k}^{M,l},x_{1:k}^{M,l-1})\right)$.}
\item{Sample, for $i\in\{1,\dots,M\}$, $(X_{k+1}^{i,l},X_{k+1}^{i,l-1})|(x_k^{i,l},x_k^{i,l-1})$ using  $f^l_{\theta}\left(\cdot|(x_k^{i,l},x_k^{i,l-1})\right)$ and set $k=k+1$ and go to step 3..}
\item{Pick a single trajectory $(X_{1:T}^{i,l},X_{1:T}^{i,l-1})$ using the normalized weights and return this and $\check{p}^M(y_{1:T})$.}
\end{enumerate}
\caption{Delta Particle Filter}
\label{alg:pf_l}
\end{algorithm}

\begin{algorithm}
\begin{enumerate}
\item Input: $L$ the target discretization level and $\{N_l\}_{l=0}^L$ the number of MCMC iterations.
\item Generate $N_0$ samples using Algorithm \ref{alg:mcmc_0}.
\item For $l\in\{1,\dots,L\}$ independently, run Algorithm \ref{alg:mcmc_l} for $N_l$ iterations.
\item Return the following approximation of $\Pi_L(\varphi)$ 
\begin{equation}\label{eq:ml_approx}
\Pi_L^{N_{0:L}}(\varphi) := \Pi_0^{N_0}(\varphi) + \sum_{l=1}^{L} \left\{
\frac{\check{\Pi}_{l,l-1}^{N_l}(\{\varphi\otimes 1\}H_1)}{\check{\Pi}_{l,l-1}^{N_l}(H_1)}
-
\frac{\check{\Pi}_{l,l-1}^{N_l}(\{1\otimes\varphi\}H_2)}{\check{\Pi}_{l,l-1}^{N_l}(H_2)}
\right\}.
\end{equation}
\end{enumerate}
\caption{Main Approach}
\label{alg:main_alg}
\end{algorithm}

\section{Mathematical Results}\label{sec:math}

The objective of this section is to provide meaningful bounds in terms of $(\Delta_0,N_0),\dots,(\Delta_L,N_L)$:
$$
\mathbb{E}\left[\left(\Pi_L^{N_{0:L}(\varphi)}-\Pi(\varphi)\right)^2\right]
$$
where the expectation operator is associated to the randomness needed to compute \eqref{eq:ml_approx}
(i.e.~on a suitably defined probability space, of which the details are omitted for brevity). We critically assume in this section that the Markov chain samplers defined in Algorithms \ref{alg:mcmc_0} and \ref{alg:mcmc_l}  are initialized from the stationary distributions (this will exist under the assumptions given below); that is in Step 2.~in both 
Algorithms \ref{alg:mcmc_0} and \ref{alg:mcmc_l} starts from samples from the stationary distribution. This assumption is discussed in both \cite{jasra_bpe_sde,jasra_cont}.

We begin by considering the two discretization schemes in Algorithms \ref{alg:rkm}-\ref{alg:rkm1} and 
 these are the only two that are considered in our analysis. To that end we will make some assumptions that are needed for the analysis. Below for a function $\varphi:\mathsf{X}\rightarrow\mathbb{R}^p$ for some $p\geq 1$, we write as $\mathcal{C}_b^k(\mathsf{X})$ as the collection of
functions on $\mathsf{X}$ to $\mathbb{R}^p$ that are
$k$ times continuously differentiable which are bounded and have bounded derivatives. For $\sigma_{\theta}(x)$,
we denote the $j^{\textrm{th}}-$column as $\sigma_{\theta}^{\cdot,j}(x)$, $j\in\{1,\dots,d\}$ and the associated $d\times d$ matrix of first derivatives in $x$ as $\nabla\sigma_{\theta}^{\cdot,j}(x)$.

\begin{hypD}\label{hyp:d2}
\begin{enumerate}
\item{For  every fixed $\theta\in\Theta$, $(\mu_{\theta},\sigma_{\theta})\in\mathcal{C}_b^6(\mathbb{R})\times\mathcal{C}_b^6(\mathbb{R})$. Denote for $j\in\{0,\dots,6\}$
$$
C_{\theta}^j = \max\left\{\sup_{x\in\mathbb{R}}\Bigg|\frac{\partial^j \mu_{\theta}}{\partial x^j}(x)\Bigg|,\sup_{x\in\mathbb{R}}\Bigg|\frac{\partial^j \sigma_{\theta}}{\partial x^j}(x)\Bigg|\right\}
$$
when $j=0$ the `zeroth' derivative is the function itself. Then, in addition:
$$
\sup_{\theta\in\Theta}\max\{C_{\theta}^0,\dots,C_{\theta}^6\} <+\infty.
$$
}
\item{For each $(x,\theta)\in\mathbb{R}^d\times\Theta$ and every $(i,j)\in\{1,\dots,d\}^2$ we have
$$
\nabla\sigma_{\theta}^{\cdot,j}(x)\sigma_{\theta}^{\cdot,i}(x) = \nabla\sigma_{\theta}^{\cdot,i}(x)\sigma_{\theta}^{\cdot,j}(x).
$$
}
\end{enumerate}
\end{hypD}

\begin{hypD}\label{hyp:d3}
\begin{enumerate}
\item{$d=1$.} 
\item{For  every fixed $\theta\in\Theta$, $(\mu_{\theta},\sigma_{\theta})\in\mathcal{C}_b^4(\mathbb{R})\times\mathcal{C}_b^4(\mathbb{R})$. In addition:
$$
\sup_{\theta\in\Theta}\max\{C_{\theta}^0,\dots,C_{\theta}^4\} <+\infty.
$$
}
\item{We have that for  every $(x,\theta)\in\mathbb{R}\times\Theta$
$$
\sigma_{\theta}(x)\frac{\partial\mu_{\theta}}{\partial x}(x) - \mu_{\theta}(x) \frac{\partial\sigma_{\theta}}{\partial x}(x)
-\frac{1}{2}\sigma_{\theta}(x)^2\frac{\partial^2\sigma_{\theta}}{\partial x^2}(x) = 0.
$$
}
\end{enumerate}
\end{hypD}

The assumption (D\ref{hyp:d2}) will relate to using Algorithm \ref{alg:rkm} as the discretization and is essentially an application of \cite[Theorem 3 (ii)]{rum}. One remark here is that on the basis of \cite[Theorem 1]{bur_bur}, we have inferred that smoothness in the statement of  \cite[Theorem 3]{rum} is the condition $(\mu_{\theta},\sigma_{\theta})\in\mathcal{C}_b^6(\mathbb{R})\times\mathcal{C}_b^6(\mathbb{R})$. The assumption (D\ref{hyp:d3}) means that we use Algorithm \ref{alg:rkm1} as the discretization scheme and is the one used in \cite[Theorem 2 (ii)]{rum}.

To continue our discussion, we will need some assumptions on the Markov chain samplers in Algorithms \ref{alg:mcmc_0} and \ref{alg:mcmc_l} as well as the likelihood function $g_{\theta}(y|x)$. These assumptions are used in \cite{jasra_bpe_sde}. We denote the Markov transition kernels (one step) that are described in 
Algorithms \ref{alg:mcmc_0} and \ref{alg:mcmc_l} as $K_0$ and $K_l$ ($l\in\mathbb{N}$) respectively. We shall denote the state-spaces of the transition kernels as $\mathsf{W}_l$, $l\in\mathbb{N}_0=\mathbb{N}\cup\{0\}$. The stationary distributions of $K_l$ are denoted as $\nu_l$. The spaces and stationary distributions are deliberately abstract as they are quite complicated (see for instance \cite{andrieu}).
For a function $\varphi:\mathsf{X}\rightarrow\mathbb{R}$, we denote by $\textrm{Lip}_b(\mathsf{X})$ the collection of bounded, measurable and Lipschitz functions.

\begin{hypA}\label{hyp:a1}
\begin{enumerate}
\item{For each $y\in\mathsf{Y}$ there exist $0<\underline{C}<\overline{C}<+\infty$ such that for each
$(x,\theta)\in\mathbb{R}^d\times\Theta$
$$
\underline{C} \leq g_{\theta}(y|x)\leq \overline{C}.
$$
}
\item{There exists a $C\in(0,1)$ such that for any $l\in\mathbb{N}_0$ there exist a probability measure $\xi$ on $\mathsf{W_l}$, such that for any $\varphi\in\textrm{Lip}_b(\mathsf{W}_l)$ and $u\in\mathsf{W}_l$
$$
\int_{\mathsf{W}_l} \varphi(u')K_l(u,du') \geq C\int_{\mathsf{W}_l}\varphi(u')\xi(du').
$$
}
\item{For each $l$, $K_l$ is $\nu_l$-reversible.}
\end{enumerate}
\end{hypA}

(A\ref{hyp:a1}) has been discussed in detail in \cite{jasra_bpe_sde} and we refer the reader there for discussion. The notation of $\eta_l-$reversibility is also defined there. To finally state our result we use the notation
$$
\mathsf{U}(N_{0:L},\Delta_{0:L},\beta) := \sum_{l=0}^L\frac{\Delta_l^{\beta}}{N_l} 
+\sum_{(l,p)\in\{1,\dots,L\}, l\neq p}\frac{\Delta_l^{\beta/2}\Delta_p^{\beta/2}}{N_lN_p} + 
\Delta_L^{\beta}.
$$
We have the following result.

\begin{prop}\label{prop:main_res}
Assume (A\ref{hyp:a1}) and either (D\ref{hyp:d1}-\ref{hyp:d2}) or (D\ref{hyp:d1},\ref{hyp:d3}). Then for any $\varphi\in\textrm{\emph{Lip}}_b(\mathbb{R}^d\times\Theta)$ there exist a $C<+\infty$ such that for any $L\in\mathbb{N}$, $N_{0:L}\in\mathbb{N}^{L+1}$:
$$
\mathbb{E}\left[\left(\Pi_L^{N_{0:L}}(\varphi)-\Pi(\varphi)\right)^2\right] \leq C~\mathsf{U}(N_{0:L},\Delta_{0:L},\beta)
$$
where $\beta=3$ under (D\ref{hyp:d1}-\ref{hyp:d2}), or $\beta=4$ under (D\ref{hyp:d1},\ref{hyp:d3}).
\end{prop}

\begin{proof}
The proof of this result, first needs the strong and weak errors associated to the SH and classical Runge Kutta methods. The strong errors are proved in  \cite[Theorem 3 (ii)]{rum} and \cite[Theorem 2 (ii)]{rum} respectively and are $\mathcal{O}(\Delta_l^3)$ and $\mathcal{O}(\Delta_l^4)$. The weak errors are easily bounded by using Jensen's inequality and the afore-mentioned strong errors. What remains is then to embed these results within the state-space model and MCMC framework of this article. This is exactly the subject that is tackled in \cite[Theorem 3.1]{jasra_bpe_sde}, \cite[Theorems 4.1, A.1, Proposition A.1]{jasra_cont}. The proofs of those articles can be repeated to conclude the result; this is omitted again for brevity.
\end{proof}

The implication of this result is then as normal in multilevel methods. Let $\epsilon>0$ be given. Then we choose $L$ so that $\Delta_L^{\beta}=\mathcal{O}(\epsilon^2)$. Then we can choose $N_l=\mathcal{O}(\epsilon^{-2}\Delta_l^{(\beta+1)/2})$ (we have $\beta\geq 2$ in our case) as the cost of the MCMC per iteration is $\mathcal{O}(\Delta_l^{-1})$. Substituting these choices into the upper-bound in Proposition \ref{prop:main_res} gives a MSE of 
$\mathcal{O}(\epsilon^2)$ and the cost to achieve this is 
$\mathcal{O}\left(\epsilon^{-2}\right)$
for the two discretization schemes, which is the optimal cost achievable with Monte Carlo methods. We note that this cost is also achievable by the Milstein method when implementable and we will perform direct comparisons to ascertain if the constant for using SRK methods is smaller (and hence an improvement) at least numerically; this is the topic of the next section. In the case of a single level method with one of the above discretization schemes, the cost is of $\mathcal{O}\left(\epsilon^{-2(1+\tfrac{1}{\beta})}\right)$ to achieve an MSE of $\mathcal{O}(\epsilon^2)$; the improvement of using multilevel methods falls versus when using the E-M method for both single and multilevels.
In that scenario, typically the single level method has a cost of $\mathcal{O}(\epsilon^{-3})$ (for MSE of $\mathcal{O}(\epsilon^2)$) and the multilevel is often $\mathcal{O}(\epsilon^{-2}\log(\epsilon)^2)$ for the same MSE.

\section{Numerical Results}\label{sec:numerics}

In this section, we verify Proposition \ref{prop:main_res} by showing that the cost to compute the multilevel approximation in \eqref{eq:ml_approx} is proportional to MSE$^{-1}$ for both discretization schemes discussed above. In addition to that we show that in one dimension, these methods are faster than the Milstein method although they have the same rate. We provide three examples: Geometric Brownian motion (GBM) process in 1d and in 3d and a nonlinear diffusion process in 2d.
The unknown parameter to be inferred is denoted by $\theta$. We will compute the estimator $\Pi_L^{N_{l_0:L}}$ for some $l_0<L$.

\subsection{Models}
We consider three models. We make a minor modification to the observation regime: we shall take the observation times on a regular grid of spacing $\delta=1/K$, for $K\in\mathbb{N}$ and $T=1$.

\subsubsection{Model 1: 1D GBM}
\begin{align*}
dX_t &= e^\theta~ X_t~ dt + \sigma~ X_t~ dW_t, \quad X_0 = x_0, ~~t\in [0,1]\\
Y_k&|X_{k\delta} \sim \mathcal{N}(\log X_{k\delta}, \tau^2), \quad k \in \{1,\cdots,K\}, ~K\in \mathbb{N}, ~~ \delta = 1/K,
\end{align*}
where $\mathcal{N}(m,\tau^2)$ denotes the Normal distribution with mean $m$ and variance $\tau^2$. Here $(X_t, \theta, \sigma) \in  \mathbb{R}_+^3$, and $\{W_t\}_{t\in [0,1]}$ is a one-dimensional Brownian motion. We assume that $e^\theta < \sigma^2/2$. We consider a Gaussian prior $\theta \sim \mathcal{N}(-1.4,0.2)$.
We set $\sigma = 0.66$, $x_0 = 0.7$, $K=120$, $\tau^2 = 0.1$ and the true parameter value from which the data is generated to $\theta^{\star}=-1.8971$. 
\subsubsection{Model 2: 3D GBM}
\begin{align*}
dX_t &= e^\theta~ X_t~ dt + \sigma~ X_t~ dW_t, \quad X_0 = x_0, ~~t\in [0,1]\\
Y_k&|X_{k\delta} \sim \mathcal{N}(\log X_{k\delta}, \tau^2~I_3), \quad k \in \{1,\cdots,K\}, ~K\in \mathbb{N}, ~~ \delta = 1/K,
\end{align*}
where $(\theta, \sigma) \in \mathbb{R}_+^2$, $X_t \in \mathbb{R}_+^3$, $I_3$ is the 3$\times$3 identity matrix, and $\{W_t\}_{t\in [0,1]}$ is a three-dimensional Brownian motion. The log here is taken element-wise. We assume that $e^\theta < \sigma^2/2$.
We set $\sigma = 0.66$, $x_0 = (0.7,0.7,0.7)$, $K=120$, $\tau^2 = 0.1$ and the true parameter value from which the data is generated to $\theta^{\star}=-1.8971$. We consider a Gaussian prior $\theta \sim \mathcal{N}(-1.4,0.04)$.

\subsubsection{Model 3: 2D Nonlinear Diffusion Process}
\begin{align*}
dX_{1,t} &= -\theta~ X_{1,t}~ dt +  \frac{\sigma}{\sqrt{1+X_{1,t}^2}}~ dW_{1,t}, \\
dX_{2,t} &= -\theta~ X_{2,t}~ dt +  \frac{\sigma}{\sqrt{1+X_{2,t}^2}}~ dW_{2,t}, \qquad (X_{1,0}, X_{2,0}) = (x_{1,0}, x_{2,0}), ~~t\in [0,1]\\
Y_k|&X_{k\delta} \sim \mathcal{N}(X_{k\delta}, \tau^2), \quad k \in \{1,\cdots,K\}, ~K\in \mathbb{N}, ~~ \delta = 1/K.
\end{align*}
Here, $(\theta, \sigma) \in \mathbb{R}_+^2$, $\{W_{1,t}\}_{t\in [0,1]}$ and $\{W_{2,t}\}_{t\in [0,1]}$ are two one-dimensional independent Brownian motions.
We set $\sigma = 1$, $x_0 = (-1,-2)$, $K=120$, $\tau^2 = 4$ and the true parameter value from which the data is generated from to $\theta^{\star}=1$. We consider a Gaussian prior $\theta \sim \mathcal{N}(-1.2,0.1)$.

\subsection{Simulation Settings}

We set $\varphi(x_{(1:T)\delta},\theta)=\theta$. We also have $l_0=1$, $M=120$ particles (for the PF) and $N_{\text{burn}}=8000$, the burn-in period in the MCMC. The multilevel estimator $\Pi_L^{N_{l_0:L}}$ is computed for target MSEs $\epsilon^2$ in $\{2\times 10^{-3}, 2\times 10^{-3}, \cdots, 2\times 10^{-6}\}$ for the first two models and $\{2\times 10^{-2}, 2\times 10^{-3}, \cdots, 2\times 10^{-5}\}$ for the third model. For each MSE, the target discretization level is chosen as $L=2|\log_2(\epsilon)|/\beta$, with $\beta =2,3,4$ is the strong error rate of Milstein (for Model 1), SH (Models 1 \& 2) and SRK schemes (Models 1 \& 2), respectively. Since it can be easily shown that Model 3 does not satisfy assumptions D\ref{hyp:d2} and D\ref{hyp:d3}, both schemes fail to attain the expected strong error rates. This is confirmed in Figure \ref{fig:NLM_rates}, where plot the strong and weak error rates for the SH and SRK schemes when applied on Model 3. The figure shows that $\beta=2$ for both algorithms. In addition, we set $N_l=\lceil 2\epsilon^{-2}K_L\Delta_l^{(\beta+1)/2}\rceil + N_{\text{burn}}$, for $l\in \{l_0,\cdots,L\}$. The cost corresponding to each chosen MSE is given by $\mathcal{C}ost = \sum_{l=l_0}^L (N_l-N_{\text{burn}}) \Delta_l^{-1}$. The reference for each method is computed using the multilevel estimator in Algorithm \ref{alg:main_alg} at a higher level of discretization. The simulations are independently repeated $50$ times and the MSEs are computed based on the reference value of each method.

\begin{figure}[h]
\centering
\includegraphics[height=6cm,width=0.8\textwidth]{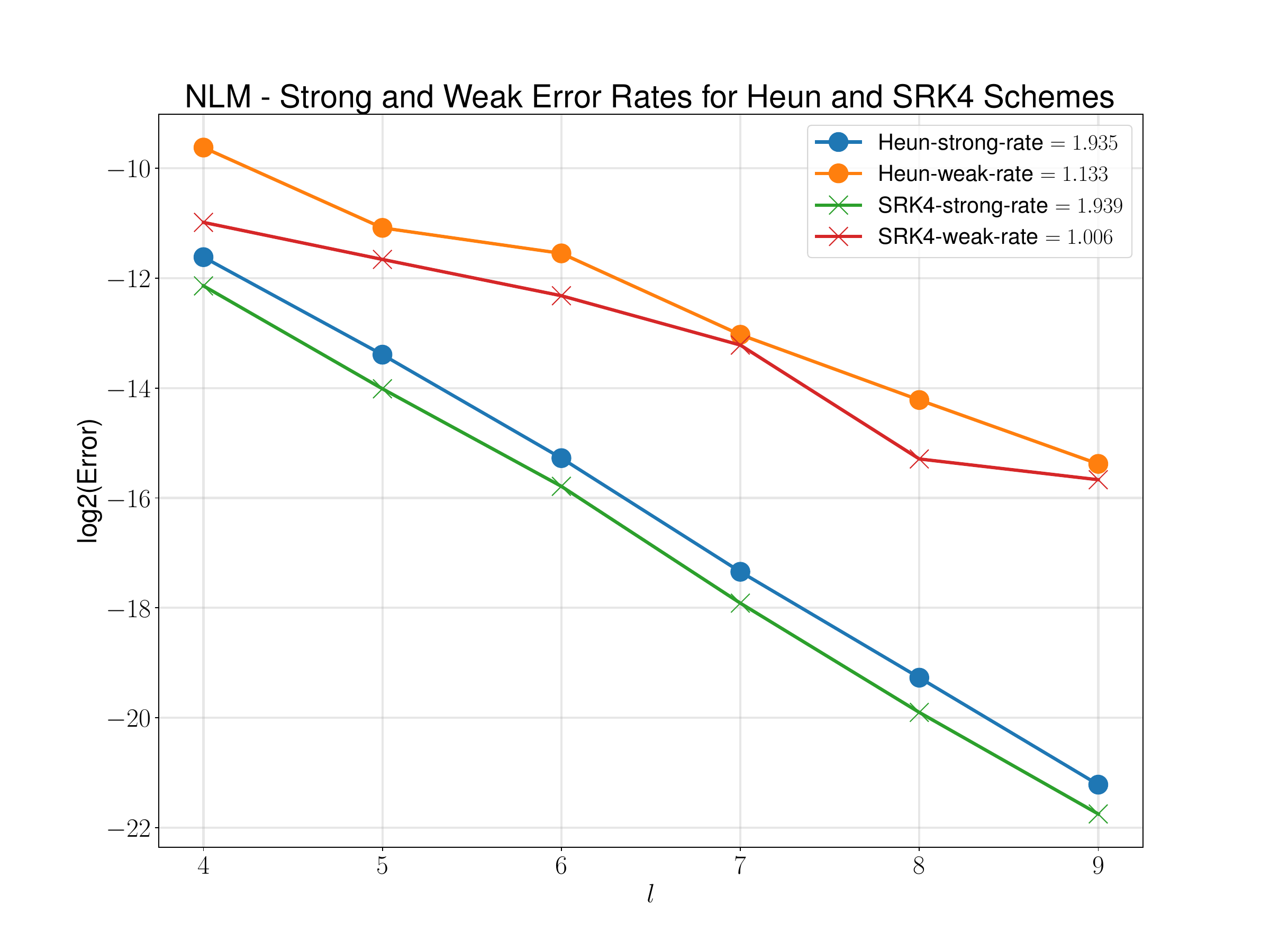}
\captionof{figure}{Strong and weak error rates of SH and classic SRK schemes when applied on the nonlinear diffusion process Model 3. The reference is the mean of 200 simulations at discretization level $L=18$.}
\label{fig:NLM_rates}
\end{figure}
\subsection{Simulation Results}
The main results of the cost vs. MSE are shown in Figures \ref{fig:cost_vs_mse1}-\ref{fig:cost_vs_mse3}. Since $\beta >1$ for all schemes considered when applied on the models above, the cost is of order MSE$^{-1}$, which is verified with our simulations. Additionally, while these methods have similar cost-error rates, the plots illustrate that in the cases of  Models 1 \& 2, the multilevel estimator presented in Algorithm \ref{alg:main_alg} achieves faster convergence when paired with the classical SRK method compared to the SH method, which, in turn, outpaces the Milstein method when the later is used for the one-dimensional model. However, for Model 3, since the strong error rate seems to be the same for both SH and SRK (as illustraded in Figure \ref{fig:NLM_rates}) and since SRK involves more calculation steps, SH, as a result, will have a lower cost; this is verified in Figure \ref{fig:cost_vs_mse3}.

\begin{figure}[h]
\centering
\includegraphics[height=6cm,width=0.8\textwidth]{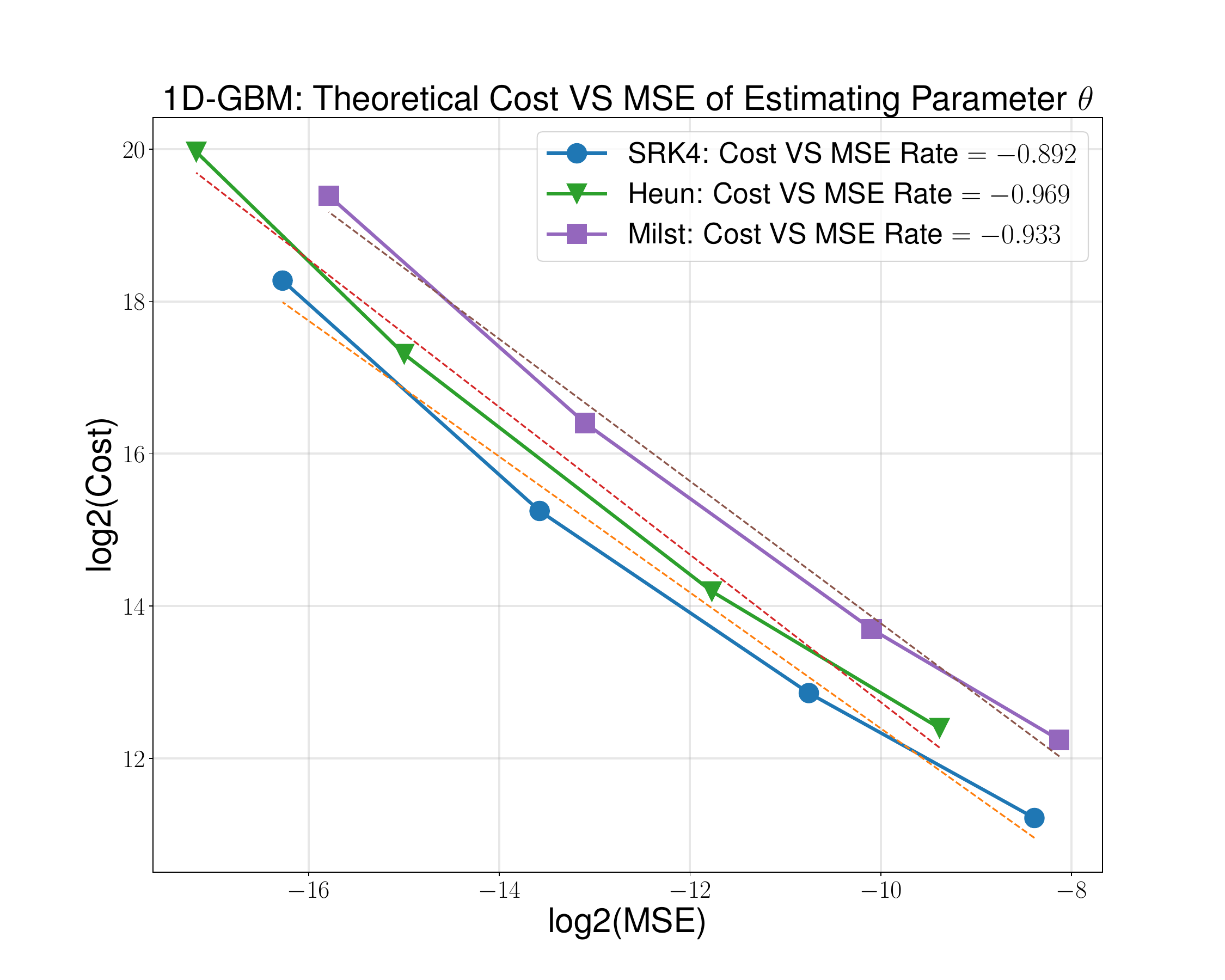}
\captionof{figure}{Cost vs. MSE for the multilevel estimator of parameter $\theta$ in Model 1.}
\label{fig:cost_vs_mse1}
\end{figure}

\begin{figure}[h]
\centering
\includegraphics[height=6cm,width=0.8\textwidth]{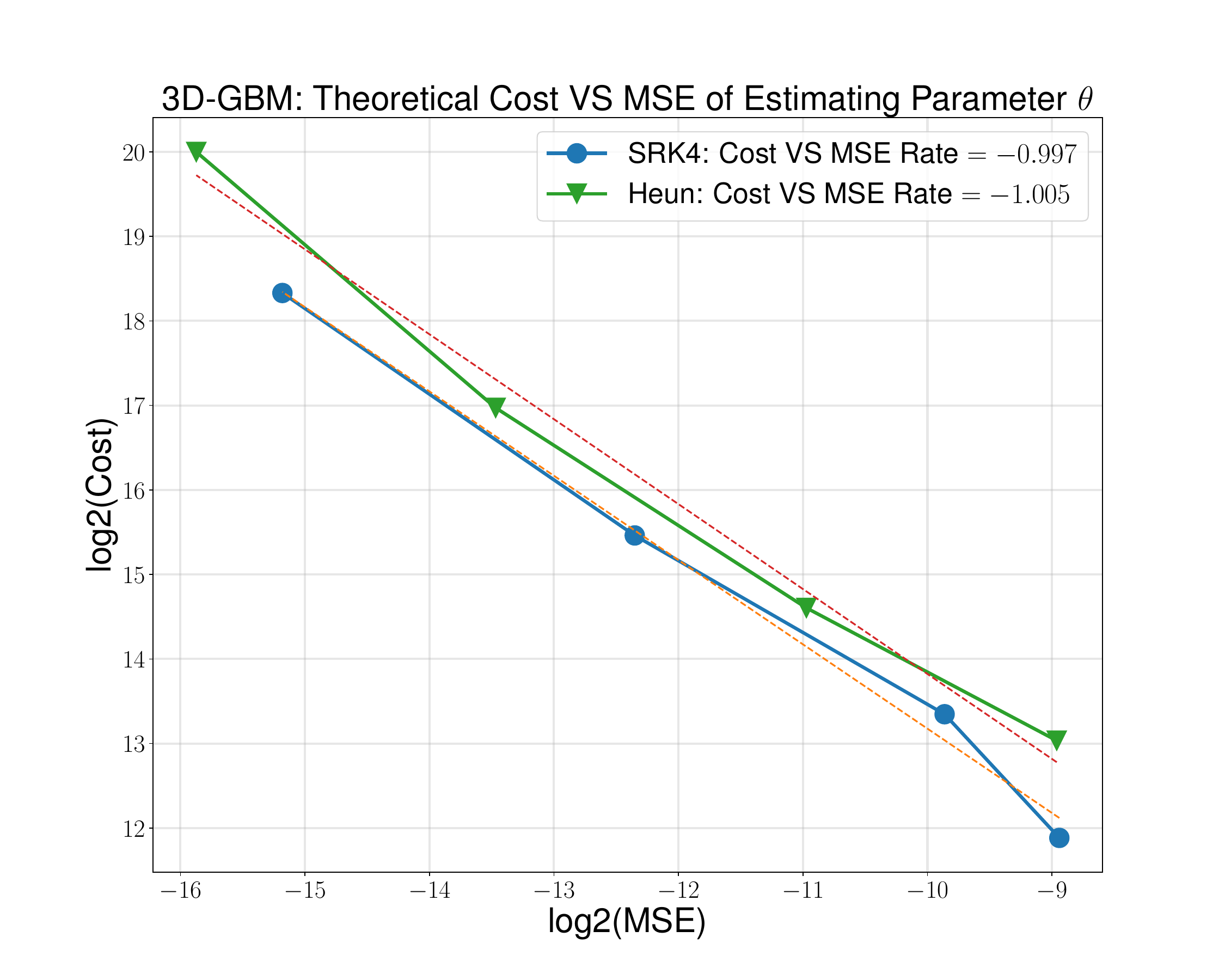}
\captionof{figure}{Cost vs. MSE for the multilevel estimator of parameter $\theta$ in Model 2.}
\label{fig:cost_vs_mse2}
\end{figure}

\begin{figure}[h]
\centering
\includegraphics[height=6cm,width=0.8\textwidth]{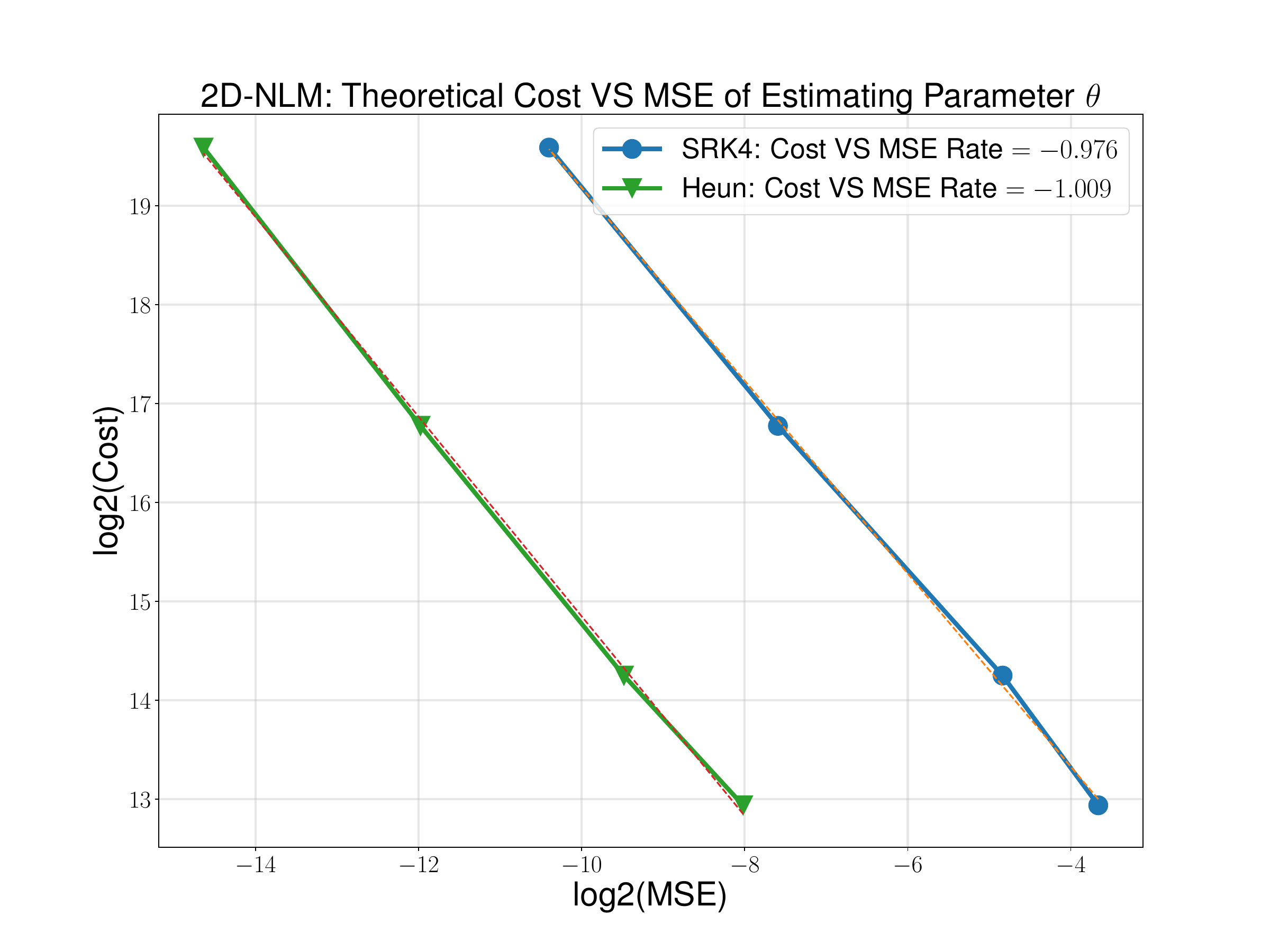}
\captionof{figure}{Cost vs. MSE for the multilevel estimator of parameter $\theta$ in Model 3.}
\label{fig:cost_vs_mse3}
\end{figure}

\subsubsection*{Acknowledgements}

AJ \& HR are supported by KAUST baseline funding.

\end{document}